\documentclass{llncs}
\usepackage{url} 

\usepackage{amsmath}
\usepackage{amssymb}
\usepackage{booktabs}
\usepackage{tikz}
\usepackage{multirow}
\usepackage{enumitem}

\let\phi=\varphi
\renewcommand{\phi}{\varphi}

\begin{document}

\title{Succinct data structures for representing equivalence classes\thanks{Work done while the first and the third authors were on sabbatical at the University of Waterloo, Canada}}

\author{Moshe Lewenstein\inst{1} \and J. Ian Munro\inst{2} \and Venkatesh Raman\inst{3}}

\institute{Department of Computer Science, Bar Ilan University, Israel \\
\email{moshe@cs.biu.ac.il}
\and
Cheriton School of Computer Science, University of Waterloo, Canada \\
\email{imunro@uwaterloo.ca}
\and
The Institute of Mathematical Sciences, Chennai, India \\
\email{vraman@imsc.res.in}
}

\maketitle

\begin{abstract}
Given a partition of an $n$ element set into equivalence classes, we consider time-space tradeoffs for representing it to support the query that asks whether two given elements are in the same equivalence class. 
This has various applications including for testing whether two vertices are in the same connected component in an undirected graph or in the same strongly connected component in a directed graph.

We consider the problem in several models.
\begin{itemize}
\item
Concerning labeling schemes where we assign labels to elements and the query is to be answered just by examining the labels of the queried elements (without any extra space):
if each vertex is required to have a unique label, then we show
that a label space of $\sum_{i=1}^n \lfloor {n \over i} \rfloor$ is necessary and
sufficient. In other words, $\lg n + \lg \lg n + O(1)$ bits of space are necessary and sufficient for representing each of the labels. This slightly strengthens the known lower bound and is in contrast to the known necessary and sufficient bound of $\lceil \lg n \rceil$ for the label length, if each vertex need not get a unique label.
\item
Concerning succinct data structures for the problem when the $n$ elements are to be uniquely assigned labels from label set $\{1,\ldots, n\}$,  we first show that $\Theta(\sqrt n)$ bits are necessary and sufficient to represent the equivalence class information. This space includes the space for implicitly encoding the vertex labels.
We can support the query in such a structure in $O(\lg n)$ time in the standard word RAM model.

We then develop structures where the queries can be answered
\begin{itemize}
\item
in $O(\lg \lg n)$ time using $O(\sqrt n \lg n/\lg \lg n)$ bits, and
\item
in $O(1)$ time using $O({\sqrt n} \lg n)$ bits of space.
\end{itemize}
\end{itemize}
En route, we provide an interesting method to compute the integer nearest to the square root of integers up to $n$ using a table look up. 
We believe that this method can be of independent interest.

We also develop a dynamic structure that uses $O(\sqrt n \lg n)$ bits to support equivalence queries and unions in $O(\lg n/\lg \lg n)$ worst case time or $O(\alpha (n))$ expected amortized time where $\alpha (n)$ is the inverse Ackermann function.

\end{abstract}

\section{Introduction and Motivation}
We look at the following problem. Given a partition of an $n$ element set into equivalence classes, preprocess it, assigning a unique label to each element, to obtain a data structure with minimum space to support the following query: given two elements, determine whether they are in the same equivalence class.  We call the query an `equivalence query'.
This is a fundamental data structure problem that has various applications including for testing whether two vertices are in the same connected (or strongly connected) component in an undirected (or directed) graph. We study the problem in the context of succinct data structures.
Designing succinct (or space efficient) data structures has been an area of interest in theory and practice motivated by the need to store large amount of data. See~\cite{BCHM,FM,MN,RRR07,Brodnik-Munro} for succinct representations of dictionaries, trees, arbitrary graphs and partially ordered sets.

%

We address the time-space tradeoff for representing an equivalence class and answering the equivalence query in a couple of models.
Katz, Katz, Korman and Peleg~\cite{sicomp} introduced the notion of labeling schemes whereby every node of the graph is assigned a (not necessarily distinct) label and the required query is to be answered by just looking at the labels of the query elements. They showed that $\Omega (k \lg n)$
\footnote{We use $\lg n$ to denote $\log_2 n$}
is a lower bound of the length of the label to answer `$k$-connectivity queries', for $k$ up to polylogarithmic in $n$. For $k=1$ (which is the case for the problem in this paper), this lower bound is $\lceil \lg n \rceil$
and hence the scheme that simply assigns all elements of an equivalence class a single label that is distinct from the labels of other equivalence classes, is optimal in this model. However, in some situations (for example when one wants to support other graph operations including adjacency relations) we may want to give unique labels to each vertex.
Our first result is that in this case, we need a label space of $\sum_{i=1}^n \lfloor n/i \rfloor$, and we show that this number of labels is also sufficient. We also give an encoding scheme that uses the optimal $\lg n + \lg \lg n + O(1)$ bits for the labels. The encoding scheme is similar to the one in~\cite{SIDMA}, but our lower bound is stronger, and more importantly we establish an exact tight bound for the label space. This result is discussed in Section~\ref{labscheme}.

Then, in Section~\ref{sds}, we give succinct data structures for the problem in the model where the labels can be freely reassigned (but need to be unique and in the range $1$ to $n$), and the query can be answered by looking at a small space data structure. We first observe that the information theoretic lower bound to represent the equivalence class information is
$\Omega (\sqrt n)$ bits, and we provide a scheme using $O(\sqrt n)$ bits in which the query can be answered in $O(\lg n)$ time. In the rest of the section, we develop a data structure where the query can be answered in constant time albeit using $O(\sqrt n \lg n)$ bits of space. In Section~\ref{dynamic}, we develop methods that also support merge operation on the equivalence classes using asymptotically the same space, as fast as other known non-space efficient structures.

These structures operate in the standard word RAM model with a word size of $w=\Omega (\lg n)$~\cite{wordRAM} where multiplication and shifts can be performed in constant time. Furthermore, our succinct structures modify the initial labels of the elements to an implicit labeling scheme. We discuss applications and limitations of this approach in Section~\ref{concl}.

%
%

%
%
%
%
%
%
\section{Labeling scheme with unique labels for elements}
\label{labscheme}
In the problem, which we call the {\em direct equivalence queries} problem,  each element is to be given a unique label, and the equivalence query is to be answered by computing directly from the two labels. It is known~\cite{SIDMA} that $\lg n + \Theta (\lg \lg n)$ bits of space are necessary and sufficient to represent the labels. We strengthen the bound to
$\lg n + \lg \lg n + \Theta (1)$. The encoding that achieves this bound is similar to the one in \cite{SIDMA}, we provide it for completeness, but our lower bound establishes a tight bound on the label space.
We first prove the following theorem.
\begin{theorem}
\label{labelrange}
Let a partition of an $n$ element set into equivalence classes be given as input to the direct equivalence queries problem. Then a label space of $\sum_{i=1}^n \lfloor n/i \rfloor$ is necessary and sufficient.
\end{theorem}

\begin{proof}
Our key observation for the sufficiency is that the $i$-th largest equivalence class contains at most $\lfloor n/i \rfloor$ elements.
For the upper bound, we simply assign labels from the set of integers in the range $[ \sum_{j=1}^{i-1} (\lfloor n/j \rfloor) + 1, \sum_{j=1}^i \lfloor n/j \rfloor ] $ for the $i$-th largest equivalence class, for $i>1$, and integers in the range $[1, n]$ for the largest equivalence class.

To show that this many labels are necessary, consider the collection of $n$ equivalence relations (partitions of an $n$ element set) as below.
The collection $C_i$ contains $i$ sets (equivalence classes) each containing $\lfloor n/i \rfloor$ or $\lceil n/i \rceil $ elements. In particular if $s$ and $t$ are the sizes of two of these classes, then $|s-t| \leq 1$.

Consider the labels assigned by any labeling scheme for the above $n$ collection of equivalence relations.
Note that the labels assigned to the relation $C_1$ can be assigned to at most one class of each of $C_i, i=2$ to $n$. This happens because every pair of elements are in the same equivalence class in $C_1$, and hence we will have a conflict (to answer the equivalence query looking only at the labels) if these labels are assigned to more than one class of $C_i, i=2$ to $n$.
Now remove $C_1$, and all classes from $C_i, i=2$ to $n$ that have been assigned the same labels as of $C_1$. Now the proof follows by repeating the above argument with the labels assigned to the elements of (the remaining classes of) $C_2$, $C_3$ up to $C_n$ in that order. \qed
\end{proof}

To answer the equivalence query in the above labeling scheme, given an integer label $x$, we need to find the largest $i$ such that $\sum_{j=1}^{i-1} \lfloor n/j \rfloor < x$. In order to support this query in constant time, we modify the labeling scheme slightly (and use space slightly suboptimal, up to lower order terms).  We first order the equivalence classes in non-increasing order of their sizes. We give them labels, say $1$ to $c$ where $c$ is the number of classes. Within each class, we give an arbitrary ordering of the elements. Then the label for an element $x$ is given by a pair ($i, j$) where
$i$ is the label of the class to which the element belongs,  and $j$ is its `rank' in the class numbered $i$. As the $i$-th largest equivalence class contains at most $\lfloor n/i \rfloor$ elements, the label $j$ can be represented using $\lceil \lg \lfloor n/i \rfloor \rceil$ bits of space.
The label $i$ is represented using $\lceil \lg i \rceil$ bits. As the size of the representation of $i$ is not fixed, we need to store information to find the `break point' between $i$ and $j$. Hence we `prefix' the label $(i, j)$ by storing the length of $i$ in binary,
using $\lceil \lg \lceil  \lg n \rceil \rceil$ bits.  The equivalence query can easily be answered by looking at the first component ($i$) of the label in constant time. The number of bits used for a label is $\lceil \lg \lceil \lg n \rceil \rceil + \lceil \lg i \rceil + \lceil \lg \lfloor n/i \rfloor \rceil$ which is
at most $\lg n + \lg \lg n + 2$.

>From Theorem~\ref{labelrange}, $\lceil \lg \sum_{i=1}^n \lfloor n/i \rfloor \rceil = \lceil \lg (n \ln n - O(n)) \rceil$ bits are necessary for the label length. Thus we have
\begin{theorem}
Given a partition of an $n$ element set into equivalence classes,  we can assign to each of the elements a label of $\lg n + \lg \lg n + 2$ bits such that the equivalence query can be answered in constant time by looking only at the labels. In this model, $\lg n + \lg \lg n - \Theta (1)$ bits are necessary to represent the labels.
\end{theorem}
\section{Succinct Data Structures}
\label{sds}
Now we move on to designing data structures, where the labels of the $n$ elements can be freely reassigned, but they need to be unique and in the range $1$ to $n$.  The queries can be answered by looking at an augmented data structure. 
We are interested in time and space efficient data structures. We first assign an implicit ordering of the elements. Each element gets a label according to this ordering, and the queries are answered by looking at these labels and an augmented data structure.


First, we address the question of how much space is required to capture the given equivalence class information.
The information theory lower bound for the representation is given by the number of partitions of an $n$ element set into equivalence classes, which is the same as the number of partions of $n$, which by the Hardy-Ramanujan formula~\cite{HR} is
asymptotically ${1 \over 4n \sqrt 3} e^{(\pi \sqrt {2n \over 3})}.$
Hence the information theoretic lower bound for space to represent the equivalence class information is given by $\pi \sqrt {2n/3} \lg e - \lg n + O(1)$ which is $\Theta (\sqrt n)$.

Now, to design space efficient data structures, let $c$ be the number of classes, $s_i, i=1$ to $k$ be the distinct sizes of the classes, and let $n_i$ be the number of classes of size $s_i$ in the given equivalence class. Key to our structure is ordering the classes in non-decreasing order of $\gamma_i= s_i n_i$. I.e. $s_i n_i \leq s_{i+1} n_{i+1}$, for $i=1$ to $k-1$.
We first make the simple observations that
\begin{equation}
\label{one}
\sum_{i=1}^k s_i n_i = n, \sum_{i=1}^k n_i = c  ~{\rm and}~ s_i n_i \geq i ~{\rm for} ~ i=1 ~ {\rm to}~ k
\end{equation}
The last inequality follows as the $i$-th smallest $s_i$ value is at least $i$.
It follows from these observations that that $k \leq c$ and $k \leq \sqrt {2n}$.
\subsection{Structure using $O(\sqrt n)$ bits}
\label{optspace}
Here, we design a structure that uses $O(\sqrt n)$ bits of space to represent the equivalence class information, and can support equivalence query in $O(\lg n)$ time.
Our primary structure consists of two sequences:
\begin{itemize}
\item
the sequence $s$ that consists of $\delta_i = s_i n_i - s_{i-1} n_{i-1}$, $i=1$ to $k$, where $s_0n_0$ is defined to be $0$ and
\item
the sequence $m$ that consists of $n_i, i=1$ to $k$.
\end{itemize}

Each element in these sequences is represented in binary (using respectively 
$1 + \lceil \lg (\delta_i+1) \rceil$ and $1+ \lceil \lg (n_i +1) \rceil$ bits). As the lengths of each element in the sequence vary, we store two other sequences that `shadow' the two primary sequences.
The first one $\psi$ has a $1$ at the starting point of each element in the sequence $s$ and $0$ at other positions. Similarly, the second one $\rho$ stores a $1$ at the starting point of elements of the sequence $m$, and $0$ at other positions. We also store a select structure (see for example~\cite{RRR07,haste})
on these two sequences $\psi$ and $\rho$ to identify the $1$s quickly. The space occupied by each of these two sequences is clearly the same as that occupied by the two primary sequences, plus lower order terms.

The first sequence gives an implicit ordering of the elements, i.e. the
elements in the first $n_1$ classes are assigned label values $1$ to $s_1n_1$, the elements of the next $n_2$ classes are assigned the next $s_2n_2$ label values and so on.
%

We first claim that the space occupied by these four sequences is $O(\sqrt n)$ bits. We first show the following Lemma.
If any $\delta_j=0$, then we account for $1$ bit for its representation and 
as $k$ is $O(\sqrt n)$, this doesn't affect the claimed bound; so assume that $\delta_j \geq 1$ for all $j$ in the sum below.
\begin{lemma}
$\sum_{j=1}^k \lg \delta_j$ is $O(\sqrt n)$ where each $\delta_j$ (as defined above) is at least $1$.
\end{lemma}

\begin{proof}
We use the following claim to achieve the desired bound.\\

\noindent
{\bf Claim:} For an integer $1 \leq i \leq n$, the number of $j$'s such that $\delta_j \geq i$ is at most $\sqrt {2n/i}$.

\noindent
{\bf Proof of claim:} Let $\delta_{j_t} \geq i$, for some $t=1$ to $b$. Then
$s_{j_t} n_{j_t} \geq ti$, and hence
$$ \sum_{t=1}^b ti \leq \sum_{t=1}^b s_{j_t} n_{j_t} \leq n$$
from which it follows that $b(b+1)/2 \leq n/i$ or $b \leq \sqrt {2n/i}$ which proves the claim. \qed

\noindent
>From the claim, it follows that (by breaking the $\delta$ values into ranges of powers of two -- i.e. those between $2^{p-1}$ and $2^p$ for various values of $p$)

$$\sum_{j=1}^k \lg \delta_j  \leq \sum_{p=1}^{\lceil \lg n \rceil} (\sqrt {2n/2^{(p-1)}}) p = 2 \sqrt {n} \sum_{p=1}^{\lceil \lg n \rceil} {p \over 2^{p/2}}$$
which is $O(\sqrt n)$.
\qed
\end{proof}

A similar proof shows that $\sum_{j=1}^k \lg n_j$ is $O(\sqrt n)$.
This is because if $n_j =i$ for some $j$, then $s_j n_j \geq ji$, and a claim as above follows for the number of $j$'s with $n_j=i$ as well.
Thus we have a structure to represent the equivalence class information that uses $O(\sqrt n)$ bits. \\

\noindent
{\bf Implementing the equivalence query}
Now, given an element labeled $x$, the equivalence class it belongs to is determined by first finding the predecessor $p(x)$ of $x$, which is $max \{j| \sum_{i=1}^j s_i n_i < x\}$. Given two elements $x$ and $y$, if $p(x)$ and $p(y)$ are not the same, then $x$ and $y$ are not in the same equivalence class.

If $p(x)$ and $p(y)$ are the same, then we know that $x$ and $y$ are in classes that have the same sizes, but it is still not clear whether they are in the same equivalence class. They are in the same equivalence class if and only if
$\lceil {(x-\sum_{i=1}^{p(x)} s_i  n_i)/n_{p(x)+1}} \rceil $ and
$\lceil {(y-\sum_{i=1}^{p(y)} s_i  n_i)/n_{p(y)+1}} \rceil $ are the same. To compute the $n_i$ value for some $i$, we simply look for the $i$-th and $(i+1)$-st $1$ in the sequence $\rho$ (using the select data structure on $\rho$) which gives the starting position and the length of the representation of $n_i$ in the sequence $m$.

Now in order to support the predecessor queries in a reasonable amount of time, we store more: we simply store the $\sum_{j=1}^i s_j n_j$ for every value of $i$ which is a multiple of $\lg n$. This takes $O(\sqrt n)$ bits.

Now $p(x)$ can be obtained by doing a binary search for $x$ on these partial
sum values  $\sum_{j=1}^i s_j n_j$ for every value of $i$ which is a multiple of $\lg n$.
Once an $O(\lg n)$ range of the predecessor is found, the actual predecessor value is found by doing a linear search on the delta values in this range. As before, the lengths and the starting positions of the $\delta$ values can be found using the select substructure on the sequence $\psi$.  Thus we have
\begin{theorem}
Given a partition of an $n$ element set into equivalence classes, it can be stored using $O(\sqrt n)$ bits such that the equivalence query can be answered in $O(\lg n)$ time. Furthermore, $\Omega (\sqrt n)$ is the minimum number of bits necessary to store the equivalence class information on an $n$ element set.
\end{theorem}
\subsection{Faster, Space-Efficient Methods}
\label{opttime}
Here we develop a data structure where the equivalence query can be answered in constant time albeit using $O(\sqrt n \lg n)$ bits of space.

Our initial representation consists of storing
\begin{itemize}
\item
the sequence $\sum_{j=1}^i s_j n_j$, $i=1$ to $k$, and
\item
the sequence $n_i, i=1$ to $k$,
\end{itemize}

where each number in each sequence is represented in binary using $\lceil \lg n \rceil$ bits.
As before, the first sequence gives an implicit ordering of the elements. That is, the $s_1 n_1$ elements of the first $n_1$ classes form the {\it first} $s_1 n_1$ elements and so on.
The total space used by the four sequences is
at most $2 \sqrt {2n} \lceil \lg n \rceil$ bits.
As discussed earlier, to answer the equivalence queries, we essentially have to support predecessor queries in the sequence $\sum_{j=1}^i s_j n_j$, $i=1$ to $k$.

%
%

A simple binary search can support the predecessor query in $O(\lg k)$ time.
A y-fast trie~\cite{Willard} can support the predecessor query in $O(\lg \lg n)$ time. As in the scheme of the previous subsection, we could store the complete partial sums and a y-fast trie structure storing every $\lg \lg n$-th element  in the partial sum sequence and store the $\delta$ values for the remaining elements of the sequence. This will help us find a range of $\lg \lg n$ for the predecessor in $O(\lg \lg n)$ time. Within the range, we can do a sequential search for the predecessor using the $\delta$ values. As the delta values require only $O(\sqrt n)$ bits of space,
we have
\begin{theorem}
Given a partition of an $n$ element set into equivalence classes, it can be stored using
$O(\sqrt n \lg n/\lg \lg n)$ bits such that the equivalence query can be answered in $O(\lg \lg n)$ time.
\end{theorem}

A fully indexable dictionary~\cite{RRR07} with the improved redundancy of ~\cite{haste} can support the predecessor query in constant time albeit using $O({\sqrt n}^{1+\epsilon})$ bits of space.
However, we argue below that the predecessor can be supported in constant time using an additional $O(\sqrt n \lg n)$ bits using the fact that our sequence satisfies the last inequality in equation (\ref{one}) and hence is special.
In addition to the two sequences above, we store an array $A$ of
$ \lceil {\sqrt {2n}} \rceil $ pointers, where $A[i] = max \{j|\sum_{t=1}^j s_t n_t \leq i(i+1)/2 \}$, for $i=1$ to $\lceil \sqrt {2n} \rceil$.
Now, we claim
\begin{lemma}
The predecessor $p(x)$ of an integer $x$ ($1 \leq x \leq n$) in the sequence $\sum_{t=1}^i s_t n_t, i=1$ to $k$ is $A[\lceil \sqrt {2x} \rceil -1]$ or
$A[\lceil \sqrt {2x} \rceil -1] -1$ or  $A[\lceil \sqrt {2x} \rceil -1] +1$.
\end{lemma}

\begin{proof}

Let $i = \lceil \sqrt {2x} \rceil - 1$, then $$x - (\sqrt x)/2 \leq i(i+1)/2 < x + (\sqrt x)/2, $$ and $$x + (\sqrt x)/2 \leq (i+1)(i+2)/2 < x + 3(\sqrt x)/2.$$

For $j = A[i] +1$,
$\sum_{t=1}^{j} s_t n_t > i(i+1)/2$ (by definition of $A[i]$). Hence $s_j n_j \geq (i+1)$ and hence $s_{j+1} n_{j+1} \geq (i+2)$ hence $\sum_{t=1}^{j+1} s_t n_t > i(i+1)/2 + i + 2 > x$ and hence $p(x) \leq j = A[i] +1$.

Let $l = p(x)$. Then $\sum_{t=1}^{l+1} s_t n_t > x$ and hence $s_{l+1} n_{l+1} \geq \lceil \sqrt {2x} \rceil -1$. Hence $\sum_{t=1}^{l+2} s_t n_t \geq x + \lceil (\sqrt {2x}) \rceil > i(i+1)/2$. Hence $A[i] \leq l+1 = p(x) + 1$ which implies that $A[i]+1 \geq p(x) \geq A[i]-1$. \qed
\end{proof}

The actual value of $p(x)$ can be computed by looking at the sum up to each of these three values. \\

\noindent
{\bf Computing Square Roots} Note that computing $\lceil \sqrt {x} \rceil$ is not a constant time operation in the standard word RAM model. The standard Newton's iterative method uses $\Theta (\lg \lg n)$ operations. We describe a space efficient method that avoids explicit computation of square roots (for the range we are interested in) by using a look up to precomputed tables. We use two tables,
one when the number of digits of $x$ (up to its most significant $1$) is odd, denoted by $O$, and one when the number of digits is even, denoted by $E$. It turns out that $O[i]$ and $E[i]$ are quite close in value, where $E[i]$ is roughly a $\sqrt{2}$ factor larger than $O[i]$.

For $i=1$ to $\lceil \sqrt {2n} \rceil$, we precompute and store in 
$E[i]$, the value of $\lceil \sqrt {i 2^{(\lceil \lg (i+1) \rceil)/2}} \rceil$
and in $O[i]$, the value of $\lceil \sqrt {i 2^{(\lceil \lg (i+1) \rceil)/2 - 1}} \rceil$. This takes $O(\sqrt n \lg n)$ bits.
Now, given an integer $i$, $1 \leq i \leq 2n$, we compute $\lceil \sqrt {i} \rceil$ as follows. Let $i = a_i 2^{\lceil (\lg i)/2 \rceil} + b_i$ where $b_i < 2^{\lceil (\lg i)/2 \rceil}$. Then,
\begin{lemma} $\lceil \sqrt i \rceil = E[a_i]$ or $E[a_i+1]$ if the number of digits in $i$ (up to its most significant $1$) is even, and is $O[a_i]$ or $O[a_i+1]$ otherwise.
\end{lemma}
\begin{proof}
As $i = a_i 2^{\lceil (\lg i)/2) \rceil} + b_i$,
$ a_i 2^{\lceil (\lg i)/2 \rceil} \leq i < (a_i+1) 2^{\lceil (\lg i)/2 \rceil}$,
and hence
$\lceil \sqrt {a_i 2^{\lceil (\lg i)/2 \rceil}} \rceil \leq \lceil \sqrt i \rceil \leq \lceil \sqrt {(a_i+1) 2^{\lceil (\lg i)/2 \rceil}} \rceil$
which is what we wanted to show.
\qed
\end{proof}

The actual value of $\lceil \sqrt i \rceil$ can be computed by squaring the values in the table and comparing them with $i$. Note that for $i \leq 2n$, $a_i \leq \lceil \sqrt {2n} \rceil$, and it can be obtained as follows: find the most significant bit, say bit $r$, mask the lower $r$ bits to keep only the higher half of them, i.e. $\lfloor {r\over 2}\rfloor$ of the bits (without the leading zeroes), and finally shifting them to the right by $\lceil {r\over 2}\rceil$. The most significant bit can be found in constant time with the standard RAM operations, see~\cite{FW}.
Thus we have
\begin{lemma}
For $1 \leq i \leq n$, $\lceil \sqrt i \rceil$ can be computed in constant time (for each $i$) using a precomputed table of $O(\sqrt n \lg n)$ bits.
\end{lemma}

Indeed using this approach to provide a seed for Newton iteration, one can compute $\lceil \sqrt i \rceil$, for $i=1$ to $n$ in time $O(\lg (1/\epsilon))$ using a table of $O((n^{\epsilon} \lg n) /\epsilon)$ bits, for any positive constant $\epsilon < 1$.  To summarize, we have
\begin{theorem}
\label{consttime}
Given a partition of an $n$ element set into equivalence classes, the partition can be represented using $O(\sqrt n \lg n)$ bits such that the equivalence query can be answered in constant time.
\end{theorem}

\section{Supporting Unions}
\label{dynamic}
Finally we discuss space efficient structures that can support merging of two classes in an equivalence relation and still support equivalence queries.
The merge operation takes two classes of the equivalence relation and merges them to obtain a new class destroying both the old ones.  We show
\begin{theorem}
\label{union-find}
Given a partition of an $n$ element set into equivalence classes, it can be represented using $O(\sqrt n \lg n)$ bits such that the equivalence query and merge queries can be supported in $O(\lg n/\lg \lg n)$ worst case time. In fact, using the same space, equivalence query can be supported in $O(\alpha (n))$ amortized time and merge queries can be supported in $O(\alpha (n))$ expected amortized time, where $\alpha (n)$ is the inverse Ackermann function.
\end{theorem}

\begin{proof}
The primary structure we maintain is the one as in the proof of Theorem \ref{consttime}. To support merge operations, we maintain an auxiliary structure that captures the merges that have happened until $O(\sqrt n)$ sets have merged. During this time, the original labeling of the elements is maintained. After $O(\sqrt n)$ merges have happened, the entire data structure is reconstructed with relabeling of the elements.

In the following, we represent an equivalence class (that has been involved in a merge) by the smallest element (label) in the class.
The auxiliary structure contains
\begin{itemize}
\item
a forest $F$ of rooted trees with the nodes having the label of set that has been involved in the merges since the previous relabeling of elements. Each node has a (parent) pointer to the node containing the label of the class to which it has been merged. If it is the root node, this pointer is a NIL pointer. We also keep a counter to indicate the number of edges (i.e. the number of merges that have happened since the last relabelling) in the forest. 
\item
A succinct data structure $M$ for the labels of all the sets that have been involved in merges since the previous relabeling of the elements, where insert and membership can be supported. In this structure, we store the parent pointer along with the element (or a nil pointer if the element is a root of 
$F$).  We also store a marker bit with each element indicating whether or not its a leaf in $F$.
\end{itemize}
To support the equivalence query, we proceed as in the proof of Theorem \ref{consttime} in the primary structure. If the two elements of the query are in the same class, then we answer affirmatively. Otherwise, we check the auxiliary structure to see whether those two sets had been since merged. This is done by
\begin{enumerate}
\item
first checking whether these two set labels are present in the auxiliary structure $M$. If either one is not present in $M$, then we return that the two elements of the query are in different equivalent classes.
\item
If both lables are present in $M$, then we follow through their parent pointers in $F$, and find the root of the trees they belong to. If both the roots are the same, then we report that  both elements are in the same equivalence class and otherwise report that they are in different equivalence classes.
\end{enumerate}
The complexity of these operations is dominated by the membership query in $M$, and the find query in $F$.

To support the merge query, first let us assume that the number of merges that have happened (which can be determined by looking at the counter in $F$) since the last relabelling is at most $c \sqrt n$ for some fixed constant $c$. We update the counter in $F$ after every merge until the counter reaches $c \sqrt n$.
Let $x$ and $y$ be the labels of the sets to be merged. If either of them is not in $M$, then we insert them into $M$, and then create a new tree containing that element for each (or either) of them.
Then we simply make the root of the tree containing $x$ or $y$ the child of the the root of the tree containing the other, as dictated by the union-find algorithm, and update the parent pointer of the node whose parent changed, in $M$. We also mark or unmark the leaf marker with $x$ or $y$ in $M$.

The time is dominated by the time to do find and union in $F$, and to perform membership and insert in $M$.

If $c \sqrt n$ real merges have happened since last relabelling (as indicated by the counter in $F$), the relabelling is reconstructed with the new classes and their sizes, and the auxiliary structures is cleared. For this, first
we compute the sizes of the newly constructed sets (which are at the roots of the forest $F$) as follows: Initialize their sizes to $0$. 
We scan through elements in $M$ until we reach the first leaf (as indicated by the leaf marker). 
The size of the set of the root will be increased by the sizes of all sets along the path from the leaf to the root. Simultaneously, the `original' sizes of these sets are updated by updating the $n_i$ values and the $s_i n_i$ values.

For this, we follow the parent pointer in the path up to the root by performing the following operations for each node $x$ along the path (including the leaf). 
We will initialize the increment value to $0$. 
Using Theorem~\ref{consttime}, find the set containing $x$, and the size of the set $s_i$. Decrement $n_i$ by $1$, decrease $\gamma_i = s_i n_i$ by $s_i$.
Add $s_i$ to the increment value (to be added to the size of the set in the root). 

Once we reach the root, we increment the size of the class represented by the root node by the increment value (which is the sum of the sizes of the sets in the leaf to root path).
Now we go to the structure $M$, continue to find the next leaf and repeat. Once we are done with a leaf to root path, we will mark that leaf in $M$ as visited, so that we don't repeat that path. We are done once we have traversed all the leaves in $M$.

Thus in $O(\sqrt n)$ time, we have computed the sizes of the new sets (roots in $F$), and
updated the values of the `old' $\gamma$ values along with the new number ($n_i$s) of classes contributing to the $\gamma$ value.
Now we need to sort and merge the `old' $\gamma$ values with the new set sizes, except we need to determine whether the new set sizes already exist. For this purpose, we first organize the $\gamma $ values (including the new set sizes whose $n_i$ values are $1$ temporarily) based on the $s_i$ values. The set size for a $\gamma_i$ is obtained from its $n_i$ 
value. After we organize the $\gamma$ values based on the set sizes, all $\gamma$ values with the same set sizes (at most two of them, one from the old and one from the new) will be together. So in another scan, by appropriately updating the $\gamma$ and the $n_i$ values, we can ensure that there is only one $\gamma$ value for each set size. 
Now all we need to do is to sort the $\gamma$ values and appropriately move the $n_i$ values. Sorting based on the sizes and the $\gamma$ values can be done in $O(\sqrt n)$ time by a radix sort~\cite{esa2007} as the $s_i$ and $\gamma$ values are at most $n$.
Thus, the restructuring results in $O(\sqrt n)$ time, for an amortized cost of $O(1)$. At this point, we can clear the auxiliary structures $M$ and 
$F$. The entire step can be deamortized using strandard tricks~\cite{overmars,raman} by starting the reconstruction $\sqrt n$ steps before, while maintaining the old and new structures during this process.

By using a fusion tree~\cite{FW,fusion1} for the insert and membership structure $M$ and the union-find data structure of~\cite{ABAR99,Blu86,Smi90} for $F$, 
the worst case bounds of the theorem follow.  
By using a dynamic perfect hashing scheme~\cite{DPH} for $M$ instead of the fusion tree, the amortized bounds of the theorem follow.
\qed
\end{proof}
\section{Conclusions}
\label{concl}
We have discussed time-space tradeoffs for the fundamential problem of supporting equivalence queries.
Our first result is an establishment of a tight bound for the label space required for the elements to answer equivalence query by just looking at the labels.

Then we showed that one can represent an equivalence relation on $n$ elements using $O(\sqrt n)$ bits of space, which is a constant factor of the information theoretically optimum number of bits required. 
Our scheme allows an implicit labeling of elements and supports
equivalence queries in $O(\lg n)$ time. Improving this to constant time is an interesting open problem, though we could achieve constant time using $O(\sqrt n \log n)$ bits.

We also developed a dynamic structure where the merge operation can also be supported as fast as the standard union-find structures using $O(\sqrt n \lg n)$ bits. Our main contribution is on the time-space tradeoffs for representing equivalence queries using clever use of several known structures. 

Not withstanding our claim in Theorem~\ref{union-find} that we can support unions and finds on an $n$ element set using $O(\sqrt n \lg n)$ bits and at the same asymptotic time as the best known (not so space efficient) structures, we don't know of a direct way to apply them for supporting static or (incremental) dynamic connectivity queries on graphs. This is because the original labels of the elements are modified to obtain our space efficient structure. Hence to support the user queries, we need to store the permutation that maps the user labels to our labels or update the user of the labelling (in the latter case, the user herself can answer the connectivity query). This is particularly important in our dynamic structure that supports union, as every so often, the labels are recomputed.

An example setting where our structures can be applicable is as follows. Consider a distributed environment where multiple processors are performing some intensive computation. These processors are space constrained, and each processor receives a request with a label from two different processors to perform some computation. Assume that the application requires the processor receiving this request to first determine whether the two labels are in the same equivalence class to perform the computation. So it does that using our small space union-find structure. 
In the case of dynamic (merge) queries, all processors must be aware of all the merges happening at any of the processors to perform their own local computation in case the labels get changed. Alternatively we can assume synchrony and communicate the relabelling after every change. 

One useful query in this scenario that can be supported easily by our structure is the following. Each of the processors may have a small subset of labels about which they are particularly interested in.
When a request for an equivalence query comes in, the processor may also want to know whether there is any element in its interest set that is in the same equivalence class as the query element.
This can be supported in constant time by maintaining a succinct 
membership structure (as in~\cite{Brodnik-Munro}) for the classes in which the elements in its interest set belongs to, in addition to our succinct union-find structure.

Finally, given that union-find is a fundamental structure for representing equivalence classes, we feel that our structure and approach will find applications in other scenarios we haven't imagined. \\

\noindent
{\bf Acknowledgement} The second author gratefully acknowledges the discussions he had with Tetsuo Asano which initiated work on the problem.

\end{document}